\documentclass{article}[10pt]

\usepackage{jucs2e}
\usepackage{epsfig}
\usepackage{graphicx} 
\usepackage{url}

\begin{document}

\title{Finger Indexed Sets: New Approaches\footnote{Preliminary version of this paper was presented in Advances in  Informatics,  LNCS 3746, pp 134-144, Volos, Greece, November 2005
}}

\author{{\bfseries Spyros Sioutas}\\
   (Ionian University Corfu, Greece\\
   sioutas@ionio.gr)\\
}
\maketitle

\begin{abstract}
In the particular case we have insertions/deletions at the tail of a given set S of $n$ one-dimensional elements, we present a simpler and more concrete algorithm than that presented in [Anderson, 2007] achieving the same (but also amortized) upper bound of $O(\sqrt{logd/loglogd})$ for finger searching queries, where $d$ is the number of sorted keys between the finger element and the target element we are looking for. Furthermore, in general case we have insertions/deletions anywhere we present a new randomized algorithm achieving the same expected time bounds. Even the new solutions achieve the optimal bounds in amortized or expected case, the advantage of simplicity is of great importance due to practical merits we gain.  
\end{abstract}


\begin{category}
E.1, E.5
\end{category}

\section{Introduction}

By finger search we mean that we can have a \emph{finger} pointing at a sorted key $x$ when searching for a key $y$. Here a finger is just a reference returned to the user when $x$ is inserted or searched for. The goal is to do better if the number $d$ of sorted keys between $x$ and $y$ is small. Also, we have finger updates, where for deletions one has a finger on the key to be deleted, and for insertions, one has a finger to the key after which the new key is to be inserted. 
In the comparison-based model of computation Ramman [Raman, 1992] has provided optimal bounds, supporting finger searches in  $O(logd)$ time while supporting finger updates in constant time. On the pointer machine, Brodal et al. [Brodal, 2003] have shown how to support finger searches in $O(logd)$ time and finger updates in constant time. Finally, Anderson and Thorup presented in [Anderson, 2007] optimal bounds on the RAM; namely $O(\sqrt{logd/loglogd})$ for finger search with constant finger updates in worst-case. This optimal solution is also very complicated and as a consequence not at all practical. 

In this paper, assuming that the insert/delete operations occur at the tail of set S, we present a new algorithm based on an implicit Nested Balanced Distributed Tree (BDT), which handles finger-searching queries in optimal amortized (and not worst-case) time ($O(\sqrt{logd/loglogd})$)but also in a simpler manner than that presented in [Anderson, 2007]. Consequently, our method is much easier to be implemented.

In general case we have insertions/deletions anywhere we present a new simple randomized algorithm based on application of oblivious on-line simple pebble games [Raman, 1992] upon a new 2-level hybrid data structure where the top-level structure is a Level-Linked Exponential search tree [Beam, 2002] and the bottom level are buckets of sub-logarithmic size. Our new randomized method results in the following complexities: $O(\sqrt{logd/loglogd})$ and $O(1)$ in expected case for finger searching and update queries respectively. 

In the following section we review the preliminary data structures. In section 3 we review in detail an extended outline of our new solution in special case we have insertions/deletions at the tail of given set. In section 4 we study the general case we have insertions/deletions anywhere constructing a randomized algorithm achieving the same optimal expected time bounds. In section 5 we conclude.

\section{Preliminary Data Structures }

\subsection{Precomputation Tables}

Ajtai, Fredman and Komlos have shown in [Ajtai, 1984] that subsets of the integers $\left\{1, \ldots, n\right\} $ of size polylogarithmic in $n$ can be maintained in constant time so that predecessor queries (find the largest $i \in S$ such that $i \leq x$) can be performed in constant time. In fact, their result is in the cell probe model of computation; however, on a logarithmic word size RAM their functions can be represented by tables that can be incrementally precomputed at a cost of $O(1)$ worst-case time and space per operation. The data structure occupies space that is linear in the size of the subset.

\subsection{Fusion Tree}

At STOC'90, Fredman and Willard [Fredman, 1990] surpassed the comparison-based lower bounds for sorting and searching using the features in a standard imperative programming languages such as $C$. Their key result was an $O(logn/loglogn)$ time bound for deterministic searching in linear space. The time bounds for dynamic searching include both searching and updates. Since then much effort has been spent on finding the inherent complexity of fundamental searching problems.

\subsection{Amortized Exponential Search Tree}

In 1996, Anderson [Anderson, 1996] introduced exponential search trees as a general technique reducing the problem of searching a dynamic set in linear space to the problem of creating a search structure for a static set in polynomial time and space. The search time for the static set essentially becomes the amortized search time in the dynamic set. From Fredman and Willard [Fredman, 1990], he got a static structure with $O(\sqrt{logn})$ search time, and thus he obtained an $O(\sqrt{logn})$ time bound for dynamic searching in linear space. Obviously the cost for searching is worst-case while the cost for updates is amortized.

\subsection{Beam-Fich (BF) structure}

In 2002 Beame and Fich [Beam, 2002] showed that $O(\sqrt{logn/loglogn})$ is the exact worst-case complexity of searching static set using polynomial space. Using the above mentioned exponential search trees, they obtained a fully dynamic deterministic search structure supporting search, insert, and delete in $O(\sqrt{logn/loglogn})$ amortized time. The BF structure can use randomization (for rehashing) in order to achieve $O(loglogN)$ expected update time, where $N$ is the universe. The amortized operations are very simple to be implemented in a standard imperative programming language such as $C$ or $C++$.

\subsection{Worst - Case Exponential Search Tree}

Finally, in 2007, Anderson and Thorup [Anderson, 2007] developed a worst-case version of exponential search trees, giving an optimal $O(\sqrt{logn/loglogn})$ worst-case time bound for dynamic searching. They also extended the above result to finger searching problem, achieving the same optimal time bound $O(\sqrt{logd/loglogd})$. The rebuilding operations are also very complicated and very difficult to be implemented in a standard imperative programming language such as $C$ or $C++$.

\section{A special case of finger searching}

We use as a base structure a Balanced Distribution Tree (BDT). In such a tree the degree of the nodes at level $i$ is defined to be $d(i)=t(i)$, where $t(i)$ indicates the number of nodes present at level $i$. This is required to hold for $i\geq1$, while $d(0)=2$ and $t(0)=1$. It is easy to see that we also have $t(i)=t(i-1)*d(i-1)$, so putting together the various components, we can solve the recurrence and obtain for $i\geq1$: $d(i)=2^{2^{i-1}}$, $t(i)=2^{2^{i-1}}$. One of the merits of this tree is that its height is $O(loglogn)$, where $n$ is the number of elements stored in it. 

We consider the case we have only insertions/deletions at the end of the set $S$, for example $insert(y)$ or $delete(y)$ such 
as $y > maximum \left\{x_i\in S\right\}$, $1\leq i \leq n$ or $y = maximum \left\{x_i\in S\right\}$, $1\leq i \leq n$ respectively. We build our structure by repeating the same kind of BDT tree-structure in each group of nodes having the same ancestor, and doing this recursively.

This structure may be imposed through another set of pointers (it helps to think of these as different color pointers). The innermost level of nesting will be characterized by having a tree-structure, in which no more than two nodes share the same direct ancestor. Figure 1 illustrates a simple example (for the sake of clarity we have omitted from the picture the links between nodes with the same ancestor).

\begin{figure}[h]
	\begin{center}
		\includegraphics[scale=0.60]{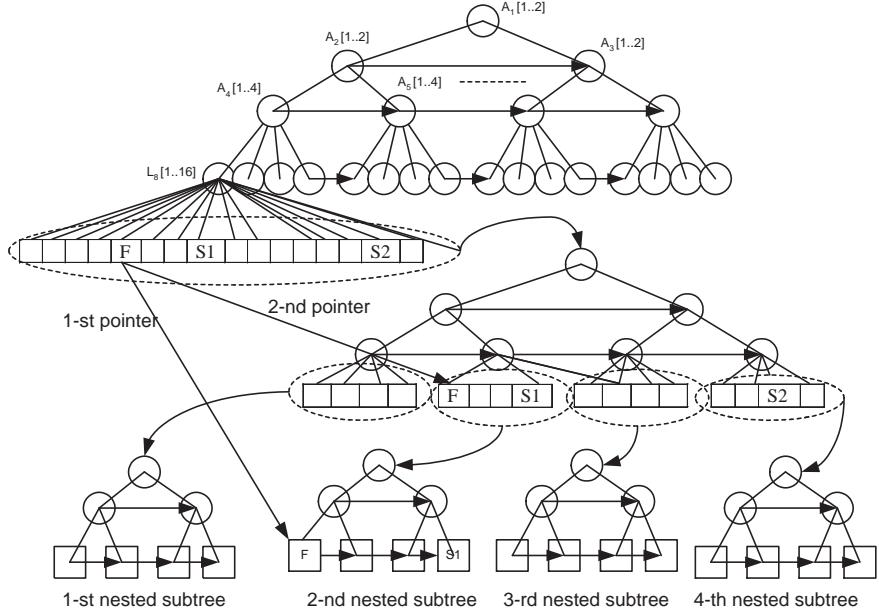}
	\end{center}

	\caption{The Level-linked leaf-oriented nested BDT tree}
	\label{fig:bdt}
\end{figure}

Thus, multiple independent tree structures are imposed on the collection of nodes inserted. Each element inserted contains pointers to its representatives in each of the trees it belongs. 

We need now to determine what will be the maximum number of nesting trees that can occur for $n$ elements. Observe that the maximum number of nodes with the same direct ancestor is $d(h-1)$. Would it be possible for a second level tree to have the same (or bigger) depth than the outermost one? 
This would imply that $\sum_{j=0}^{h-1} t(j) > d(h-1)$

As otherwise we would be able to fit all the $d(h-1)$ elements within the first $h-1$ levels. But we need to remember that $d(i)=t(i)$, thus $d(h-1)+ \sum_{j=0}^{h-1} d(j) < d(h-1)$

This would imply that the number of nodes in the first $h-2$ levels is negative, clearly impossible. Thus, the second level tree will have depth strictly lower than the depth of the outermost tree. As a consequence, the maximum number of nesting of trees $k$ that we can have is itself $O(loglogn)$.

The basic intuition behind the use of BDT tree, is the reduction of the whole set of $O(n)$ elements to the appropriate subset (nested subtree of figure 1) of $O(d)$ elements. Then by applying in this subset the simple amortized solution for general searching problem presented in [Beam, 2002], we achieve an optimal amortized solution for finger searching problem. Despite the fact that the searching time complexity of our structure is amortized and not worst-case as it happens in [Anderson, 2007] solution, it's simplicity also is of great importance since we can gain many practical merits.

We equip each node(leaf) of level $i$, say $W_i$, with a searching information array $A[1\ldots d(i)]$ ($L[1\ldots d(i)]$), where $d(i)$ is the size of the array at level $i$. We organize the elements of the arrays above with the structure of Beam-Fich presented in [Beam, 2002], let's call it $BF(W_i)$. We also equip each leaf with $k=O(loglogn)$ pointers to its respective copies at nested levels (see in Figure 1 the pointers from leaf $f$). Each element of $S$ is stored at most in $O(loglogn)$ levels, so the space of structure is non-linear $O(nloglogn)$ and the update (insertion/deletion) operation is performed in $O(loglogn)$ worst-case time.
In order to achieve linear space and $O(1)$ worst-case update time we use the bucketing technique. The essence of the bucketing method is to get the best features of these two different structures by combining them into a two-level structure. The data to be stored is partitioned into buckets and the chosen data structure for the representation of each individual bucket is different from the representation of the top-level data structure, representing the collection of buckets (for similar applications of this data structuring paradigm see also [Overmars, 1982], [Tsakalidis, 1984], [Raman, 1992]). 
More specifically, we partition the elements of the set into contiguous buckets of size $O(loglogn)$, with each bucket being represented by the linear list scheme and we store the first element of each bucket in the leaf-oriented nested balanced distributed tree scheme as its representative. When an item is inserted it is appended to the tail of the list implementing the last incomplete bucket. If the size of this bucket becomes $O(loglogn)$, then a new bucket is created containing only the inserted element, and we spend further $O(loglogn)$ time, in order to insert this element into the top-level structure. We have a total of $O(n/loglogn)$ representatives, each of which must be inserted at most in $O(loglog(n/loglogn))=O(loglogn)$ nested levels. Furthermore, at each of these levels  (leaf-levels) we must update the respective $BF$ structures in $O(loglog(d(n_i)))$ worst-case time  respectively, where $d(n_i)$ is the size of the respective array $L$ , at the $n_{i}^{th}$, $1\leq n_i\leq O(loglogn)$, level of nesting. 
More precisely the dynamic $BF$ structure requires  amortized update time but this special semi-dynamic case of updating implies the following: 

\begin{enumerate}
	\item If $n<2^{log^{2}logN/logloglogN}$ then the BF structure has only one part, the simple static data structure presented in[Anderson, 1996]. In this case we must execute a number of partial rebuilding operations at the right subtrees only of the whole structure, ensuring always that these subtrees have size at least $\frac{n}{2*\left\lceil n^{4/5}\right\rfloor}\pm 1$  and at most $\frac{2*n}{\left\lceil n^{4/5}\right\rfloor}\pm 1$, as follows. When an update causes a right-subtree to violate this condition, we examine the sum of the sizes of that subtree and its immediate neighbor which is always a full subtree with  $\frac{2*n}{\left\lceil n^{4/5}\right\rfloor}\pm 1$ elements, transferring the proper number of elements from the full neighbor node to the right-most one which we try to reconstruct. Until the next reconstruction we have all the time to spread incrementally the reconstruction cost, achieving $O(1)$ worst-case time.  So, for the $O(loglogn)$ levels of the tree depicted in figure 1 the total amount of update time becomes $O(loglogn)$ in worst-case.
	\item If $n\geq2^{log^{2}logN/logloglogN}$ or $\sqrt{logn/loglogn}\geq loglogN/(\sqrt{2}logloglogN)$ the $BF$ structure consists of two parts. The first part is a $x-fast$ trie of Willard [Willard, 1983] with branching factor $2k$ and depth $u$ which organizes the top $1+2*\left\lceil logu \right\rceil$ levels for a set of $s\leq n$ strings with length $u$, ($u=2(loglogN)/(logloglogN)\Rightarrow \sqrt{n}\geq u^u \geq logN $) over the alphabet $[0,2k-1]$. Intuitively the $x-fast$ trie reduces the predecessor and generally the dictionary problem from a universe of size $2^k$ to a subproblem with universe of size $2^b$, where $k=(logN)/2^{1+2 \left\lceil logu \right\rceil}\leq (logN)/2u^2 < u^{u-2}$, $\left\lfloor 2(u-1)^2-1\right\rfloor k < logN \leq b$  and $b\geq \left\lfloor 2(u-1)^2-1\right\rfloor k$ . The second part consists of the appropriate hash functions constructed for each resulting subproblem. When an insertion/deletion is occurred we have to insert/delete the appropriate hashed values. Since we investigate the special case where the updates occur at the tail only, the update of the hash functions described above can be done in $O(1)$ worst-case time.  So, for the $O(loglogn)$ levels of  the tree depicted in figure 1 the total amount of update time becomes again $O(loglogn)$ in worst-case.

\end{enumerate}

Due to the fact that $d(n_{i+1})=\sqrt{d(n_i)}$ at level $i$, the total amount of update operations at the appropriate $BF$ structures can be expressed as follows:

\small$O(loglog(d(n_1)))+O(loglog(\sqrt{d(n_1)}))+O(loglog(\sqrt{\sqrt{d(n_1)}}))+\ldots = O(loglogn)$
\normalsize
Spreading the total $O(loglogn)$ insertion cost, over the $O(loglogn)$ size of each bucket, we achieve an $O(1)$ amortized insertion cost. For the same reason as above it is easy to prove that the whole space is linear. We eliminate the amortization by spreading the time cost for the insertion of the representative over the next $O(loglogn)$ updates of bucket. Due to the fact that we have no a priory knowledge of $n$, we use the global rebuilding technique [Overmars, 1981] in order to retain the buckets in a appropriate size of $O(loglogn)$, where $n$ is the current number of elements.
The question is: has any affect to the $search(f,s)$ query the fact that the time, in which the query is performed, the incremental process and consequently the insertion of the bucket's representative in all possible nested levels, has not finished yet?  
In the following lemma we build the appropriate algorithm and we show that there is no possibility of such an affect.

\begin{lemma}
The $search^*(f,s)$ operation is correct and requires \small$O( \sqrt{(logd/loglogd)})$ \normalsize amortized time
\end{lemma}

\begin{proof}
Let's give the new $search^*(f,s)$ algorithm.
\\
$r_f$= representative of bucket in which finger $f$ belongs to
\\
$r_s$= representative of bucket in which $s$ belongs to
\\
$r_n$=representative of not full bucket
\\
\\
Procedure $Search^*(f,s)$

\begin{enumerate}

	\item Begin
	\item \hspace{12pt}If f, s belong to same bucket (full or not) or $s > r_n$ then access directly $s$ 
	\item \hspace{12pt}else  $fsearch (r_f, r_s)$ /*  this procedure follows */  
	\item End
\end{enumerate}
Procedure $fsearch (f,s)$
\begin{enumerate}
\item Begin
\item \hspace{12pt}W =Father(f)
\item \hspace{12pt}If $s < A_w[rightmost]$ then go to $L1$ /* f,s have the same parent */
\item \hspace{12pt}Else Begin
\item \hspace{40pt}Repeat
\item \hspace{40pt}W1=Father(W)
\item \hspace{40pt}If $A_{w1}[rightmost] < s < A_{neighbourw1[rightmost]}$
\\ 
/* that means f,s belong to neighbors nodes W1 and neighbourW1 respectively */
\\
 \item \hspace{40pt}then $fsearch(leftmostleaf(T_{neighbourw1}),s)$   
 \item \hspace{40pt}Until  $s < A_{w1}[rightmost]$
 \item \hspace{40pt}go to L2 
 \item \hspace{35pt}end
 \item \hspace{12pt}L1: Begin    
 \item \hspace{40pt}j:= -1, f=L[i] 
 \\
 /* Find the appropriate nested subtree such as $Father(f) \neq  Father(s)$ */
 \\ 
 \item  \hspace{40pt}Repeat
 \item  \hspace{40pt}j=j+1
 \item  \hspace{40pt}Until $s\leq A\left[\left\lfloor i DIV 2^{2^j} \right\rfloor 2^{2^j}+2^{2^j} \right]$
 \item  \hspace{40pt}Access the $(j+1)^{th}$ copy of  f ($f_{j+1}$) 
 \\
 /* by Following the $(j+1)^{th}$ pointer from finger(leaf) $f$
 \\   
 \item  \hspace{40pt}$fsearch (f_{j+1}, s)$
 \item \hspace{35pt}End
 \item \hspace{12pt}L2: Begin
 \item \hspace{40pt}j:=0
 \item \hspace{40pt}Repeat
 \item \hspace{40pt}j:=j+1
 \item  \hspace{40pt}search  for $s$ in $BF(W_j)$ structure 
 \\
 /* At each node of the  $W_1,W_2,\ldots,W_k,s$ path search for $s$ at $BF(W_1),\ldots,BF (W_k)$ structures respectively */
 \\
 \item \hspace{40pt}until s is found
 \item \hspace{30pt}end
 \item END                                                                 
                                                          
\end{enumerate}

\begin{enumerate}
	\item \textbf{$Search^*(f,s)$}: According to [Ajtai, 1984] the statement 2 requires $O(1)$ worst-case time. In statement 3 we call the procedure $fsearch(f,s)$ the complexity of which is analyzed as follows.
	\item \textbf{$fsearch(f,s)$}: When $f$,$s$ have the same parent (see $f$,$s1$ in figure 1), statement 3, we must determine the appropriate nested-subtree of $O(d)$ elements in which $f$,$s$ do not belong to the same collection. So, in repeat-loop 14-16 we execute exponential steps in order to find an appropriate value $j$ which defines the collection (of $2^{2^j}$  elements) in which the distance $d(f,s)$ belongs to and consequently the appropriate $(j+1)^{th}$ pointer from finger (leaf) $f$ to its respective copy $f_{j+1}$. Then we call recursively the same routine (statement 18). Obviously the repeat-loop 14-16 requires $O(loglogd)$ steps due to the fact that the distance $d$ between $f$ and $s$ is at least $d\geq 2^{2^j}$. From finger $f$ we have a number of $k=O(loglogn)$ pointers, so by organizing them in a structure of [Ajtai, 1984] we can access the $(j+1)^{th}$ pointer in $O(1)$ time.  If $f$,$s$ do not have the same parent we execute the repeat-loop of 5-9 statements that requires $O(loglogd)$ steps in order to find the nearest common ancestor of $f$ and $s$, $W_1=nca(f,s)$. If $f$,$s$ belong to neighbors nodes $W_1$ and $neighbourW_1$  respectively, (statement 7) we access the $neighbourW_1$ node in $O(1)$ time by  following the neighbor pointer from $W_1$ to $neighbourW_1$ and we call recursively the same search routine with new finger the left-most leaf of the $T_{neighbourW_1}$ subtree, else by executing the repeat-loop of 22-26 statements, we visit the appropriate search path $W_1,W_2,\ldots,W_r,s$ at each node of which we search for $s$ at $BF(W_i)$ structures, $1\leq i\leq r$ and $r=O(loglogd)$,in $O(\sqrt{logd(w_i)/loglogd(w_i)})$ amortized time, where $d(w_i)$ is the degree of node $w_i$.  This can be expressed by the following sum:   
	
$\sum_{i=1}^{r=O(loglogd)}\sqrt{\frac{logd(w_i)}{loglogd(w_i)}}$	    
              
Let $L_1$, $L_r$ the levels of $W_1$ and $W_r$ respectively. So, $d(w_1)=2^{2^{L_1}}$ and $d(w_r)=2^{2^{L_r}}$ 

But, $d(w_r)=O(d)$, so $L_r=O(loglogd)$. Now, the previous sum can be expressed as follows:
 
\small$\sqrt{\frac{2^{L_1}}{L_1}} + \sqrt{\frac{2^{L_1+1}}{L_1+1}} + \ldots + \sqrt{\frac{logd}{loglogd}}=\sqrt{\frac{logd}{loglogd}}$ 
 \normalsize

We denote that the recursive calls of statements 8, 18 are executed one time only (this fact stems from the pseudocode structure we used), consequently there is no reason to produce and solve the respective recurrence equation, so, very simply the total time becomes $T=O(\sqrt{\frac{logd}{loglogd}})$.
\end{enumerate}
   
\end{proof}

\section{A randomized algorithm with the same expected time bounds}

Let's give a brief description of the combinatorial pebble games we have to rely on for constructing our new solution.
\\
\textbf{Pebble Games [Raman, 1992]}: These games are played between two players, player $I$(increaser) and player $D$(decreaser) on a set of $n$ piles of pebbles, which are initially empty. These games have the following general form: the game is played in rounds, each consisting of one move from each player. Player $I$, on his move, increases the number of pebbles on of some of the piles, following which; player $D$ decreases the number of pebbles on some pile. Let $M$ be the maximum value of any variable at any point in the game. Player $I's$ objective is to maximize $M$, and player $D's$ to minimize it. Typically, player $D$ is an algorithm and player $I$ the environment.
\\
\textbf{Oblivious Pebble Games [Raman, 1992]}: In this type of game player $I$ reveals his moves one at a time to player $D$, but player $D's$ moves (and the status of the piles) are hidden from him. Player $D$ may use randomization to make his moves unpredictable to player $I$. Here we are interested either in the expected value of $M$ or in studying the tails of $M's$ distribution. Also, we typically restrict the number of moves this game is played, since, as it so happens, the longer the game is played, the more likely it is that player I will come close to approaching his performance in the on-line version of the game (for more details you can also see [Raman, 1992]).  
According to Oblivious On-line Discrete Zeroing Game [Raman, 1992] there is a D-strategy that ensures with high probability 
($p>1-n^{-a}$, for any constant $a>0$, for sufficiently large $n$) that over $n$ moves, $M \in O(cloglogn + clogc)$, where $c$ is an integer, $c>1$. This strategy is described from the following algorithm1:
\\
\textbf{Algorithm1}: Let $c>1$ an integer and $\delta_1,\ldots,\delta_n$ non-negative integers such that $\sum_{i=1}^{n} \delta_i=c$. Then player $D$, on his move, does the following:

\begin{enumerate}
	\item Picks i   $\left\{1,\ldots, n \right\}$ with probability $\delta_i /c$ and sets $x_i$ to zero
	\item Picks i such that $x_i=max_j \left\{x_j \right\}$ and zeroes $x_i$. 
\end{enumerate}

For $c=O(loglogn)$, $M \in O(log^{2}logn)$ with high probability. Based on D-strategy of Algorithm1 let's describe our randomized Algorithm2:

\textbf{Algorithm2}: Let $n$ be the maximum number of keys present in the data structure at any previous time. In a similar way with that presented in [Raman, 1992], we can show that making the buckets be of size $O(log^{2}logn)$ and using as top-level the structure of Beam-Fich presented in [Beam, 2002] with level-links suffice for our purposes, yielding a simple algorithm. We define the fullness $\Phi(b)$ of a bucket $b$ as in [Raman, 1992]:
\\
$\Phi(b)=\left|b\right| / log^{2}logn$. We will ensure that $0.5\leq ö(b)\leq 2$. 
\\
We also define the criticality of a bucket b to be 
\\  
\small$\rho(b,n)=\frac{1} {\alpha loglogn} max \left\{ 0, 0.7log^{2}logn- \left| b \right|, \left| b \right| - 1.8log^{2}logn \right\}$,
\normalsize
for an appropriately chosen constant $\alpha$. A bucket $b$ is called critical if $\rho(b,n)>0$. To maintain the size of the buckets, every $c=\alpha loglogn$ updates, we do the following:

\begin{enumerate}
	\item We check the $i^{th}$ bucket, $i \in \left\{1,\ldots,n/log^{2}logn\right\}$, with probability $\delta_i/c$ meaning that we construct a randomized set of $c=O(loglog(n/log^{2}logn))=O(loglogn)$ collections each of which has $O(n/log^{3}logn)$ buckets, we choice one of these collections randomly and finally the bucket of collection in which $\delta_i=max_j \left\{\delta_j\right\}$ updates have occurred. If this bucket has non-zero criticality we apply the rebalancing transformations of step 3.
	\item We check the most critical bucket and if it has non-zero criticality we apply the following rebalancing transformations.
	\item \textbf{Split}: if $\phi(b)>1.8$ split the bucket into two parts of approximately equal size.
	\\
\textbf{Transfer}: If $\phi(b)<0.7$ and one of its adjacent buckets $b'$ has $\phi(b')\geq 1$ then transfer elements from $b'$ to $b$.
\\
\textbf{Fuse}: If $\phi(b)<0.7$ and transferring is not possible, then fuse with an adjacent bucket $b'$.

\end{enumerate}

It is clear that when a critical bucket is rebalanced, it becomes non-critical. In addition to the time required to split/fuse buckets, a bucket rebalancing step may require $O(loglogN)$ expected time to insert/delete a bucket representative to/from the top-level tree. The top-level tree is the $BF$ structure, which supports updates in $O(loglogN)$ expected time. Since the total work to rebalance a bucket is $O(loglogN)$, we can perform it with $O(1)$ work per update spread over no more than $\alpha loglogn$ updates, where the chosen parameter $\alpha$ expressed as follows: $\alpha=O(\frac{loglogN}{loglogn})$. For every real computer application $N$ never exceeds the number $2^{64}=2^{2^6}$ , thus $\alpha$ could be considered as a constant much less than 6. So, if we can permit every bucket to be of size $\Theta(log^{2}log \hat{n})$, where $\hat{n}$  the number of current elements, we can guarantee that between rebalancing operation of top-level tree [Beam, 2002] there is no possibility for any other such operation to occur and consequently the incremental spread of work is possible. Let $p$ be a finger. We search for a key $k$ which is $d$ keys away from $p$. If $p$,$k$ belong to the same bucket of size $O(log^{2}logn)$, we can access directly the $k$ according to [Ajtai, 1984], else  we first check whether $r_k$ (representative of bucket in which $k$ belongs to) is to the left or right of $r_p$, (representative of bucket in which finger $p$ belongs to) say $r_k$ is to the right of $r_p$. Then we walk towards the root, say we reached node $u$. We check in $O(\sqrt{logd/loglogd})$ time whether $r_k$ is a descendant of $u$ or $u's$ right neighbor on the same level of $u$ or $u's$ right neighbor respectively. If not, then we proceed to $u's$ father. Otherwise we turn around and search for $k$ in the ordinary way.

Suppose that we turn around at node $w$ of height $h$. Let $v$ be that son of $w$ that is on the path to the finger $p$. Then all descendants of $v's$ right neighbor lie between the finger $p$ and the key $k$. The subtree $T_w$ is a $BF$ structure for 
$d$ elements, so, the total time bound $T$ becomes: 
\\
$T=O(\sqrt{logd/loglogd})$
\\
So, we proved the following theorem:

\begin{theorem}
There is a randomized algorithm with $O(1)$ and $O(\sqrt{\frac{logd}{loglogd}})$ expected time for update and finger searching queries respectively.
\end{theorem}

\section{Conclusions}

In this paper we focused on the finger searching problem. In special case we have insertions / deletions at the tail of a given set $S$, we presented an extended outline of a simpler algorithm than that presented in [Anderson, 2007] matching the optimal upper bound in amortized case. Finally, in general case we have insertions / deletions anywhere; we were based on a special combinatorial pebble game presented in [Raman, 1992] in order to present a simple randomized algorithm that achieves the same optimal expected bounds. Even the described solutions achieved the optimal bounds in amortized and expected case respectively, the advantage of simplicity is of great importance due to practical merits we can gain.  

\noindent {\bf Acknowledgements.}The author would like to thank the Program PYTHAGORAS, for funding the above work.

%\bibliography{nuernberg-final}
%\bibliographystyle{apa}

\end{document}